\documentclass[conference]{IEEEtran}
\IEEEoverridecommandlockouts
\usepackage{cite}
\usepackage{amsmath,amssymb,amsfonts}
\usepackage{algorithmicx}
\usepackage{algorithm}
\usepackage[noend]{algpseudocode}
\usepackage{amsthm}
\usepackage{graphicx}
\usepackage{subfigure} 
\usepackage{textcomp}
\usepackage{xcolor}
\usepackage{bbm}
\newtheorem{theorem}{\bf Theorem}
\newtheorem{lemma}{Lemma}
\newtheorem{Remark}{Remark}

\newtheorem{corollary}{\bf Corollary}
\def\BibTeX{{\rm B\kern-.05em{\sc i\kern-.025em b}\kern-.08em
    T\kern-.1667em\lower.7ex\hbox{E}\kern-.125emX}}
\begin{document}

\title{Improved Communication
Efficiency for Distributed Mean Estimation with Side Information
}
 
\author{\thanks{This work is supported   by   NSFC grant NSF61901267.
}\IEEEauthorblockN{Kai Liang\IEEEauthorrefmark{1} and Youlong Wu\IEEEauthorrefmark{1}\\}  
	\IEEEauthorblockA{\IEEEauthorrefmark{1}
	ShanghaiTech University, Shanghai, China}
\IEEEauthorblockA{\IEEEauthorrefmark{2} Shanghai Institute of Microsystem and Information Technology,
	Chinese Academy of Sciences}
\IEEEauthorblockA{\IEEEauthorrefmark{3} University of Chinese Academy of Sciences, Beijing, China}

	\{liangkai,wuyl1\}@shanghaitech.edu.cn
}

\maketitle

\begin{abstract}
In this paper, we consider the distributed mean estimation problem where the server has access to some side information, e.g., its local computed mean estimation  or the received information sent by the distributed clients at the previous iterations.  We propose a practical and efficient estimator  based  on an $r$-bit 
Wynzer-Ziv estimator proposed by Mayekar \emph{et al.}, which requires no probabilistic assumption on the data. Unlike Mayekar's work which only utilizes side information at the server, our scheme jointly exploits the correlation between clients' data and server's side information,  and also between  data of different clients.  We derive an upper bound of the estimation error of the proposed estimator. Based on this upper bound, we provide two algorithms on how to choose input parameters for the estimator. Finally,     parameter regions in which our estimator  is better than the previous one are characterized. 

\end{abstract}


\section{Introduction}
With the development of modern machine learning technology, more powerful and complex machine learning models can be trained through large-scale distributed training.  However, due to the large scale of the model parameters, in each iteration of the distributed optimization, the exchange of  information between distributed nodes incurs a huge communication load, causing the problem of communication bottleneck.  

We focus on  distributed mean estimation, which is a crucial  primitive for distributed optimization frameworks. Federated learning \cite{18} is one of such   frameworks, in which clients participating in joint training only need to exchange their own gradient information without sharing private data.  To alleviate  the communication bottleneck,   gradient compression \cite{1,2,3,4,5,6,7} and efficient mean estimator \cite{8,9,10,11,12,13,14} have  been investigated to  reduce the communication load.  Recently, \cite{20} studied distributed mean estimation with side information at the server, and proposed  Wyner-Ziv estimators that require no probabilistic assumption on the clients data.  

 In parallel, distributed source compression has been widely studied in classical information theory. For example, \cite{15} first studied the setting of lossy source compression with side information in the decoder. Channel coding can obtain practical coding for distributed source coding \cite{16,17}, but the main bottleneck lies in the expensive computational complexity of coding and decoding. 

In this paper, we study   practical schemes for  distributed mean estimation with side information at the server. The motivation  is based on the fact that the server could store publicly accessible data, and also at each iteration, the server has already received data sent by  clients at  previous iterations, which can be viewed  as side information. Rather than using   random coding with joint typicality tools such as  in \cite{15,Gapstar}, which is impractical to implement, we follow the work in \cite{20}  which   proposed a Wyner-Ziv estimator based on coset coding. Unfortunately, they only utilized the side information at the server, but failed to exploit   correlation between clients' vectors.  In fact, in many scenarios such as stochastic gradient descent, data between different clients may have a high correlation since they wish to learn a global  model.  Inspired by Wyner-Ziv and Slepian-Wolf coding, we  propose a  practical scheme based on the coset coding and jointly exploit  the side information at the server  and correlation between clients' data. Note that in our scheme we must address  an  ambiguity problem not existing in \cite{20} or in the classic Wyner-Ziv coding. In more detail, 
 since each client compresses its data and sends it to the server, the server only observes a lossy version of  clients' vectors.   This ambiguity cause mismatch information at the clients and server. Using the lossy version of clients' data at the server may even deteriorate the estimation. 


 

We summarize   our contributions  as follows: 1) We  propose a new estimator that improves the estimator  in \cite{20} by jointly exploiting the side information at the server and the correlation between clients' data; 2) We derive an upper bound of estimation error of the proposed estimator; 3)  We provide two greedy algorithms on how to  choose input parameters for the estimator,  and characterize the parameter regions in which our estimator  has a tighter upper bound than of the previous estimator. 


\section{Problem Setting}
Consider the problem of distributed mean estimation with side information, as depicted in Fig. \ref{imgModel}. The model consisting of $n$ clients and one server, where each Client $i\in[n]\triangleq \{1,\ldots,n\}$ observes data  $x_i\in\mathcal{X}\subset\mathbb{R}^d$ and the server has access to side information $\boldsymbol{y}=(y_1,\ldots,y_n)$, $y_i\in \mathcal{Y} \subset\mathbb{R}^d$, for some  alphabets $\mathcal{X},\mathcal{Y}$ and positive integer $d\in\mathbb{N}$. The server wishes to compute the empirical mean, i.e,
\begin{IEEEeqnarray}{rCl}
{\bar{x}}\triangleq \frac{1}{n}\sum_{i=1}^{n}{{x}}_i.
\end{IEEEeqnarray}
Note that  the  side information $\boldsymbol{y}$ could stem from some publicly accessible data or the server's guess of $\boldsymbol{x}\triangleq(x_1,\ldots,x_n)$ in the previous iterations. 

\begin{figure}
			\centering
			\includegraphics[width=0.33\textwidth]{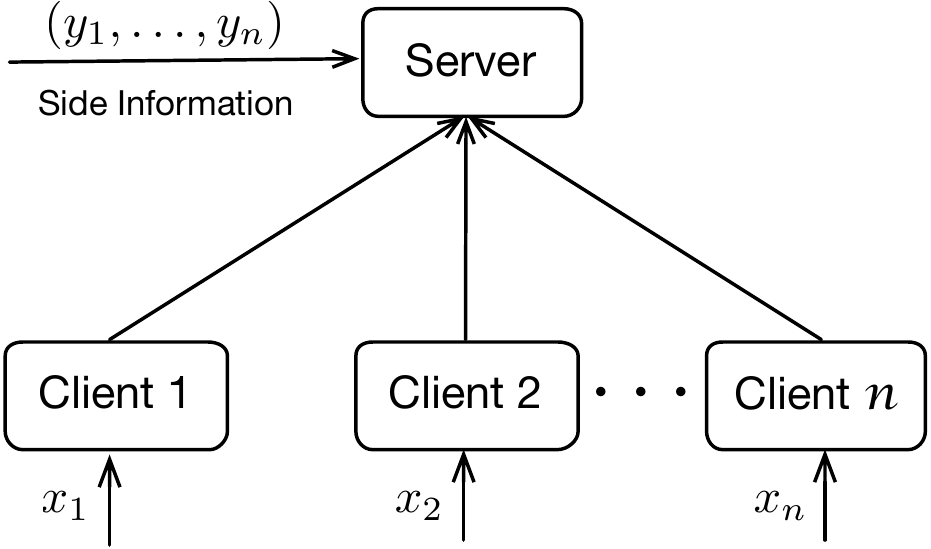} 
			\caption{Distributed mean estimation with side information}
			\label{imgModel} 
		\end{figure}


We focus on non-interactive protocols and  study the $r$-bit simultaneous message passing (SMP) protocol similar to that in \cite{20}.  The $r$-bit SMP protocol $\pi=(\pi_1,\ldots,\pi_n)$  consists of $n$ encoders $\{\Psi_i\}_{i=1}^n$ and one decoder $\Phi$, of mapping forms:
\begin{IEEEeqnarray}{rCl}
&&\Psi_i: \mathcal{X}\to\{0,1\}^r,\label{quantizer}\\
&&\Phi: \underbrace{\{0,1\}^r\times\ldots\times\{0,1\}^r}_{n~\text{times}}\times \mathcal {Y}^{n}\to {\mathbb{R}^d}.\label{decoder}
\end{IEEEeqnarray}
Each Client $i\in[n]$ uses the encoder $\Psi_i$  to encode     $x_i$ into an $r$-bit message, i.e.,
$m_i = \Psi_i(x_i,U),$
where $U$ denotes a shared randomness   known by all  the server and clients. 
 The Client  $i$ then sends the message $m_i$ to the server. Assume the message $m_i$ can be perfectly received by the server.  After receiving all messages $\textbf{m}^{(n)}=(m_1,\ldots,m_n)$, the  server uses decoder $\Phi$ to produce  $\hat{\bar{x}}$ as
 \begin{IEEEeqnarray}{rCl}
\hat{\bar{x}}=\Phi (\textbf{m}^{(n)},\boldsymbol{y},U).
\end{IEEEeqnarray}



The performance of the $r$-bit SMP protocol using protocol $\pi$  with inputs $\boldsymbol{x}$ and $\boldsymbol{y}$, is evaluated by the mean squared error (MSE), i.e.,
\begin{equation}\label{error}
	\mathcal{E}_{\pi}(\boldsymbol{x},\boldsymbol{y})\triangleq \mathbb{E}[||\hat{\bar{x}}-\bar{x}||_2^2]. 
\end{equation}

Instead of using any probabilistic assumption on input data and side information,  we use the Euclidean distance between vectors  to measure correlation  among the data and side information. More specifically, let  $ {x}_i$ and $ {y}_i$ be at most $\Delta_i$ and the distance between $ {x}_i$ and $ {x}_j$ be at most $\Delta_{ij}$, i.e.,
\begin{subequations}\label{eqDist}
\begin{IEEEeqnarray}{rCl}\label{side}
|| {x}_i- {y}_i||_2\leq \Delta_i, \forall i\in [n],\label{sidey}\\
|| {x}_i- {x}_j||_2\leq \Delta_{ij}, \forall i,j\in [n].\label{sidex}
\end{IEEEeqnarray}
\end{subequations}
Since  the  distance is symmetric with $\Delta_{ij}=\Delta_{ji}$, it's sufficient to   only consider $\Delta_{ij}$ with $i<j$. Let 
\begin{IEEEeqnarray}{rCl}
\boldsymbol{\Delta_s}\triangleq (\Delta_1,\cdots,\Delta_n), \boldsymbol{\Delta_c}\triangleq (\Delta_{12},\cdots,\Delta_{(n-1)n}).
\end{IEEEeqnarray}
 We are interested in the performance of   protocols when $\boldsymbol{\Delta_s}$ and $\boldsymbol{\Delta_c}$ are both known to clients and server. Define the     optimal $r$-bits protocol with the minimum MSE as $\pi^*$, and the corresponding MSE as $\mathcal{E}_{\pi^*}(\boldsymbol{x},\boldsymbol{y})$.  Our goal is to find practical and efficient  $r$-bits SMP protocols, and derive tighter  upper bounds on $\mathcal{E}_{\pi^*}$ than the previous results. 

 \section{Previous Work}\label{PW}
In \cite{20},  the authors proposed  a   SMP protocol based on an $r$-bit Wyner-Ziv quantizer  $Q_\textnormal{WZ}$. The quantizer $Q_\textnormal{WZ}$ contains an encoder mapping $Q_\textnormal{WZ}^\textnormal{e}$ the same as  \eqref{quantizer} and a simplified decoder mapping 
$Q_\textnormal{WZ}^\textnormal{d}:\{0,1\}^r \times \mathcal{Y}\to {\mathbb{R}^d}.$  Each Client $i\in[n]$ first uses the encoder $Q_\textnormal{WZ}^\textnormal{e}$ to encode $x_i$ and then sends the encoded message $m_i$ to the server. 
  The server uses the decoder $Q_\textnormal{WZ}^\textnormal{d}$ to produce estimate $\hat{x}_i$ as
\begin{IEEEeqnarray}{rCl}\label{decoder}
\hat{x}_i = Q_\textnormal{WZ}^\textnormal{d}(m_i,y_i),
\end{IEEEeqnarray}
and then computes the sampling means as
\begin{equation}\label{esti}
\hat{\bar{x}}= \frac{1}{n}\sum_{i=1}^{n}\hat{{x}}_i.
\end{equation}

The quantizer   $Q_\textnormal{WZ}$ achieves the following upper bound on MSE. 
\begin{theorem}[Upper bound given in \cite{20}]
For a fixed $\boldsymbol{\Delta_s}$ and $r\leq d$,  the optimal $r$-bits protocol $\pi^*$ satisfy 
\begin{IEEEeqnarray}{rCl}\label{baseline}
\mathcal{E}_{\pi^*}(\boldsymbol{x},\boldsymbol{y})
&\leq &(79\lceil \log(2+\sqrt{12\ln n})\rceil+26)(\sum_{i=1}^n \frac{\Delta_i^2d}{n^2r}),~~
\end{IEEEeqnarray}
for all $\boldsymbol{x}$ and $\boldsymbol{y}$ satisfying (\ref{eqDist}).
\end{theorem}


Now we introduce the quantizer $Q_\textnormal{WZ}$, as it is closely related to work.    Since all clients use the same quantizer,  only  the common quantizer is described.  
We first describe a modulo quantizer $Q_\textnormal{M}$ for one-dimension input $x\in\mathbb{R}$ with side information $h\in\mathbb{R}$, and then present a rotated modulo quantizer $Q_{\textnormal{M},R}(x,h)$ for $d$-dimension data.  Finally, the $r$-bit Wyner-Ziv quantizer based on $Q_\textnormal{M}$ and $Q_{\textnormal{M},R}$ is given. 
\subsubsection{Modulo Quantizer ($Q_\textnormal{M}$)}\label{cq}  Given the input $x\in\mathbb{R}$ with  side information $h\in\mathbb{R}$, the modulo quantizer $Q_\textnormal{M}$ contains parameters including a distance parameter $\Delta'$ where $|x-h|\leq \Delta'$, a resolution parameter $k\in\mathbb{N}^+$ and a lattice parameter $\epsilon$. 

Denote the encoder and decoder of $Q_{\textnormal{M}}$ as $Q^{\textnormal{e}}_{\textnormal{M}}(x)$ and $Q^{\textnormal{d}}_{\textnormal{M}}(Q^{\textnormal{e}}_{\textnormal{M}}(x),h)$, respectively. 
 The encoder $Q^{\textnormal{e}}_{\textnormal{M}}(x)$ first computes $\lceil x/\epsilon\rceil $ and $\lfloor x/\epsilon\rfloor $, and then outputs the message $Q^{\textnormal{e}}_{\textnormal{M}}(x)=m$, where
\begin{IEEEeqnarray}{rCl}
m =  \left\{ \begin{array}{llr}
(\lceil x/\epsilon\rceil\mod k), &~\text{w.p.}~x/\epsilon-\lfloor x/\epsilon\rfloor  \\
(\lfloor x/\epsilon\rfloor \mod k), &  ~\text{w.p.}~\lceil x/\epsilon\rceil- x/\epsilon
\end{array}.
\right.
\end{IEEEeqnarray}
The message $m$ has length of $\log k$ bits, and is  sent  to the decoder.  The decoder $Q^{\textnormal{d}}_{\textnormal{M}}$ produces the estimate $\hat{x}$ by finding a point closest to $h$ in the set $\mathbb{Z}_{m,\epsilon}=\{(zk+m)\cdot\epsilon:z\in\mathbb{Z}\}$. 

\subsubsection{Rotated Modulo Quantizer ($Q_{\textnormal{M},R}$)}\label{hcq}
 
  Given the input  $x\in\mathbb{R}^d$ with side information $h\in\mathbb{R}^d$ where $\|x-h\|_2\leq \Delta$, the input parameters for $Q_{\textnormal{M},R}$ include   a distance parameter  $\Delta'$, a resolution parameter $k\in\mathbb{N}^+$, a lattice parameter $\epsilon$, and a rotation matrix $R$ given by  
  \begin{equation}\label{ger}
	R=  WD/{\sqrt{d}},
\end{equation}
where $W$ is the $d\times d$ Walsh-Hadamard Matrix \cite{21}  and $D$ is a diagonal matrix with each diagonal entry generated uniformly from $\{+1,-1\}$ by using a shared randomness. After the rotation,   every  coordinate $i\in[d]$ of   $R(x-h)$, denoted by $R(x-h)(i)$, has zero mean sub-Gaussian  with  a variance factor  of $\Delta^2/d$, i.e., 
\begin{equation}\label{subga}
P(|R(x-h)(i)|\geq\Delta')\leq 2e^{-\frac{\Delta'^2 d}{2\Delta^2}}.
\end{equation}


 The quantizer $Q_{\textnormal{M},R}$ first preprocesses $x$ and $h$ by multiplying  both $x$ and $h$ with a matrix $R$,   and then  applies $Q_\textnormal{M}$ for each coordinate. Denote the encoder and decoder of $Q_{\textnormal{M},R}$ as $Q^{\textnormal{e}}_{\textnormal{M},R}(x)$ and $Q^{\textnormal{d}}_{\textnormal{M},R}(Q^{\textnormal{e}}_{\textnormal{M},R}(x),h)$, respectively.
 
 
 \subsubsection{The $r$-bit Wyner-Ziv Quantizer ($Q_\textnormal{WZ}$)}\label{rhcq}

Note that in the    quantizer $Q_{\textnormal{M},R}$ the input $x$ is encoded into  $d$ binary strings of $\log k$ bits each, leading to a total number of $d\log k$ bits. In the $r$-bit Wyner-Ziv quantizer, the encoder first encodes $x$ using the same encoder as $Q^{\textnormal{e}}_{\textnormal{M},R}(x)$, and then uses a shared randomness to select a  subset $S\subset\{1,\ldots,d\}$ of  these strings with $|S|=\lfloor r/\log k\rfloor$, and finally sends them to decoder. The decoder  uses the same decoder as $Q^{\textnormal{d}}_{\textnormal{M},R}$ to decode the entries in $S$.   Denote the encoder and decoder of $Q_{\textnormal{WZ}}$ as $Q^{\textnormal{e}}_{\textnormal{WZ}}(x)$ and $Q^{\textnormal{d}}_{\textnormal{WZ}}(Q^{\textnormal{e}}_{\textnormal{WZ}}(x),h)$, respectively. 

 \section{New Protocol and New Upper bound}
 \subsection{New Protocols}\label{chain}
 

 Note that in \eqref{decoder} only $y_i$ is used as side information to assist the estimation for $x_i$ at the server. In fact,  apart from $y_i$, the side information $\{y_j\}_{j\neq i}$ and other clients' data $\{x_j\}_{j\neq i}$ could  also be correlated to  $x_i$, and thus can be jointly utilized to reduce the transmission load. The main challenge is that $\{x_j\}_{j\neq i}$ cannot be perfectly known by the server, and thus using the estimate $\{\hat{x}_{j}\}_{j\neq i}$ as side information for $x_i$ may even deteriorate the estimation.  


 Our   protocol is based on a set of  $r$-bit new quantizers, denoted by $\{Q^{\mathcal{L}_{\Pi_i}}_\textnormal{Pro}\}_{i=1}^n$, where $\Pi_i$ denotes the $i$-th element of a permutation $\Pi$ of  $[n]$, and $\mathcal{L}_{\Pi_i}$ is a \emph{chain} parameter need to be designed and has a form of 
 $\mathcal{L}_{\Pi_i}:y_{{\Pi_i}_1}\rightarrow x_{{\Pi_i}_1}\rightarrow x_{{\Pi_i}_2}\cdots\rightarrow x_{{\Pi_i}_l},$ 
 with $x_{{\Pi_i}_l}=x_{\Pi_i}$ and $l$ being the length of chain.  

Given a set of chains $\{\mathcal{L}_{\Pi_i}:i\in[n]\}$, the input data $\{x_i\}_{i=1}^n$ are estimated in an order $x_{\Pi_1},\ldots,x_{\Pi_n}$. For the input $x_{\Pi_i}$,  the corresponding quantizer $Q^{\mathcal{L}_{\Pi_i}}_\textnormal{Pro}$  consists of an encoder the same as $Q^{\textnormal{e}}_\textnormal{WZ}(x_{\Pi_i})$, and a novel decoder  $Q^{\textnormal{d},\mathcal{L}_{\Pi_i}}_{\textnormal{Pro}}$ of mapping form 
 $$Q^{\textnormal{d},\mathcal{L}_{\Pi_i}}_{\textnormal{Pro}}: \underbrace{\{0,1\}^r\times\ldots\times\{0,1\}^r}_{i~\text{times}}\times \mathcal {Y}^{i}\to  \mathbb{R}^d, $$
 that is used to decode $x_{\Pi_i}$ as
 \begin{IEEEeqnarray*}{rCl} 
\hat{x}_{Q^{\mathcal{L}_{\Pi_i}}_\textnormal{Pro}\!,\Pi_i}\!\!=Q^{\textnormal{d},\mathcal{L}_{\Pi_i}}_{\textnormal{Pro}}
\!\big(\!Q^{\textnormal{e}}_\textnormal{WZ}(x_{\Pi_1}),\ldots,Q^{\textnormal{e}}_\textnormal{WZ}(x_{\Pi_i}), y_{\Pi_1},\ldots,y_{\Pi_i}\!\big).
\end{IEEEeqnarray*}
 Given any quantizer $Q$, denotes its  estimate for input $x_i$   as $\hat{x}_{Q,i}$.
 Here $ \hat{x}_{Q^{\mathcal{L}_{\Pi_i}}_\textnormal{Pro},\Pi_i}$ denotes the estimate for ${x}_{\Pi_i}$ when using the quantizer $Q^{\mathcal{L}_{\Pi_i}}_{\textnormal{Pro}}$ for the given chain $\mathcal{L}_{\Pi_i}$.   With a slight abuse of notation, we write $\hat{x}_{Q^{\mathcal{L}_{\Pi_i}}_\textnormal{Pro},\Pi_i}$ as $\hat{x}_{Q_\textnormal{Pro},\Pi_i}$.

In the following, we describe the quantizers  $\{Q^{\mathcal{L}_{\Pi_i}}_\textnormal{Pro}\}_{i\in[n]}$ in two steps: 1) Given  a set of chains $\{\mathcal{L}_{\Pi_i}\}_{i\in[n]}$, how to estimate $\hat{x}_{Q_{\textnormal{Pro}},\Pi_i}$, for $i\in[n]$; 2) How to select  proper chains  $\{\mathcal{L}_{\Pi_i}\}_{i\in[n]}$ to reduce the MSE. 

\subsubsection{New quantizer for some given chains $\{\mathcal{L}_{i}\}_{i\in[n]}$}

Without loss of generality, we assume that the estimation order $\Pi$ is an identity permutation, i.e., $\Pi_i=i,i\in[n]$. With this assumption,      the chain $\mathcal{L}_{\Pi_i}$can be written as 
\begin{IEEEeqnarray}{rCl}\label{chainS}
\mathcal{L}_i:y_{i_1}\rightarrow x_{i_1}\rightarrow x_{i_2}\cdots\rightarrow x_{i_l},
\end{IEEEeqnarray}
where $x_{i_l}=x_{i}$, $ i_t\in[i-1]$ for all $ t=1,\ldots,l-1$,
and the decoder $i$ already has $i-1$ estimates: $\hat{x}_{Q_\textnormal{Pro},1},\ldots,\hat{x}_{Q_\textnormal{Pro},i-1}$.

The encoder is same as $Q^e_\text{WZ}$, i.e., Client $i$  first    applies the encoder  $Q^{\textnormal{e}}_{\textnormal{M},R}(x_i)$ to encode $x_i$,  then uses a shared randomness to select a  subset $S\subset\{1,\ldots,d\}$ of  these strings with $|S|=\lfloor r/\log k\rfloor$, and finally send them to decoder.

    The decoder  $Q^{d}_{\textnormal{Pro}, i}$ chooses an element in $\mathcal{M}_i$ as the ``side" information $h$ for $Q^\textnormal{d}_{\textnormal{M},R}$, where 
    \begin{IEEEeqnarray}{rCl}
   \mathcal{M}_i= \left\{\begin{array}{llr}
   \{y_i, \hat{x}_{Q_{\textnormal{Pro}},1},\cdots,\hat{x}_{Q_{\textnormal{Pro}},i-1}\}, &~\textnormal{if}~i>1\\
   \{y_i\}, &~\textnormal{if}~i=1\\
   \end{array}\right.,
\end{IEEEeqnarray}
We emphasize that here  the ``side" information $h$ could be the estimate of other client's data, rather than the literal  side information $y_i$ used in  the Wyner-Ziv quantizer $Q_\text{WZ}$.

    
    Given the chain $\mathcal{L}_i$ in \eqref{chainS}, denote  $\Delta_{i_1}'$ and $\Delta_{i_si_{s+1}}'$ as weight parameters of subchains $y_{i_1}\rightarrow x_{i_1}$ and $x_{i_s}\rightarrow x_{i_{s+1}}$, respectively. The choices of  $\Delta_{i_1}'$ and $\Delta_{i_si_{s+1}}'$ is based on \eqref{subga}, and follows a  way similar to that  in the quantizer $Q_{\text{M},R}$. For $t\in[l]$, let  ${w_{i_t}}\triangleq \Delta_{i_1}'+\sum_{s=1}^{t-1}\Delta_{i_si_{s+1}}'$. 
    

Given a vector $v\in \mathbb{R}^d$ and a subset $S\in[d]$, let    $v(S)\triangleq (v_i:i\in S) $.  The decoder  estimates $\hat{x}_{Q_{\textnormal{Pro}},i_{t}}(S)$  as the output values in dimension $S$ of the decoder $Q^d_{\textnormal{M},R}(Q^{\textnormal{e}}_{\textnormal{M},R}(x_{i_t}),\hat{x}_{Q_{\textnormal{Pro}},i_{t-1}})$ with parameters $\Delta'=w_{i_t}$ and $h=\hat{x}_{Q_{\textnormal{Pro}},i_{t-1}}$, and estimate the values of $\hat{x}_{Q_{\textnormal{Pro}},i_{t}}([n]\backslash S)$ as those values in $y_{i_t}$, i.e.,  
\begin{subequations}\label{estimateX}
\begin{IEEEeqnarray}{rCl}
\hat{x}_{Q_{\textnormal{Pro}},i_{t}}(S)&=&Q^d_{\textnormal{M},R}(Q^{\textnormal{e}}_{\textnormal{M},R}(x_{i_t}),\hat{x}_{Q_{\textnormal{Pro}},i_{t-1}})(S),\\
\hat{x}_{Q_{\textnormal{Pro}},i_{t}}([n]\backslash S)&=&y_{i_t}([n]\backslash S).
\end{IEEEeqnarray}
\end{subequations}

By recursively using \eqref{estimateX}, we can  obtain the estimate $\hat{x}_{Q_{\textnormal{Pro}},i}=\hat{x}_{Q_{\textnormal{Pro}},i_{l}}$ for the input $x_i$. 
 

 
 In our protocol,  the parameters $h$ and $\Delta'$  depend the design of chains $\{\mathcal{L}_i\}_{i=1}^n$, $S$ is generated by the shared randomness, and $k, \epsilon$ can be  freely assigned.  When the length chain $\mathcal{L}_i$ is $l=1$, 	the quantizer $Q^{\mathcal{L}_i}_{\textnormal{Pro}}$	reduces to the Wyner-Ziv quantizer $Q_{\textnormal{WZ}}$ if the chosen chain is  $\mathcal{L}_i := y_i\rightarrow x_i$. 

    \subsubsection{Selection of Chains}
    Next we give two algorithms on how to choose chains $\{\mathcal{L}_i:i\in[n]\}$.

    \emph{Algorithm 1}: Given weight parameters $\{\Delta'_i:i\in[n]\}$ and $\{\Delta'_{ij}:i,j\in[n],i<j\}$, for Client $1$, we use the chain $\mathcal{L}_1=y_1\rightarrow x_1$ with weight $\Delta'_1$. For Client $i$, we suppose that the chains $\{\mathcal{L}_j: j\in[i-1]\}$ with weight $w_j: j\in[i-1]$ are already known. 
    
    Construct $i$ chains $\{\mathcal{L}^i(j):j\leq i\}$ as follows.
    \begin{equation}\label{g11}
    	\begin{aligned}
    	&\mathcal{L}^i(j)=\mathcal{L}_j\rightarrow x_i,~ \text{if}~ j\in[i-1],\\
    	&\mathcal{L}^i(j)=y_j\rightarrow x_j,~ \text{if}~ j=i.
    	\end{aligned}
    \end{equation} Then,  compute the weights $\{w^i(j):j\in[i]\}$ according to 
   \begin{equation}\label{g12}
    \begin{aligned}
    &w^i(j)=w_j+\Delta'_{ji},~ \text{if}~j\in[i-1],\\
   &w^i(j)=\Delta'_{j}, ~ \text{if}~ j=i.
    \end{aligned}
    \end{equation} 
   For each Client $i$,   $\mathcal{L}_i$  can be chosen from   $i$ candidate chains in $\{\mathcal{L}^i(j):j\in[ i]\}$.  
    We choose   $\mathcal{L}_i=\mathcal{L}^i(j^*) $ such that  $j^*\triangleq\arg\min_{j\in[i]}w^i(j)$. 
      We formally describe the method in Algorithm~\ref{alg4}.
 \begin{algorithm}
	\caption{Selection of Chains}\label{alg4}
	\hspace*{0.02in}{\bf Input:}
	Input $\{\Delta'_i:i\in[n]\}$ and $\{\Delta'_{ij}:i,j\in[n],i<j\}$\\
	\hspace*{0.02in}{\bf Output:}
	\emph{Chains}
	\begin{algorithmic}[1]
		\State ${Chains}\leftarrow \emptyset$
		\State ${Chains}\leftarrow {Chains}\cup\{\mathcal{L}_1=y_1\rightarrow x_1\}$
		\State ${Weights}\leftarrow \emptyset$
		\State ${Weights}\leftarrow {Weightss}\cup\{w_1=\Delta'_1\}$
		\For{$2\leq i\leq n$}
		    \State Generate $\{\mathcal{L}^i(j):j\leq i\}$ as (\ref{g11})
		    \State Compute $\{w^i(j):j\leq i\}$ as (\ref{g12})
		    \State Compute $j^*=\arg\min_{j\in[i]}w'(j)$
		    \State $\mathcal{L}_i\leftarrow \mathcal{L}^i(j^*)$,  $w_i\leftarrow w^i(j^*)$
		    \State ${Chains}\leftarrow {Chains}\cup\{\mathcal{L}_i\}$
		    \State ${Weights}\leftarrow {Weights}\cup\{w_i\}$
		\EndFor
		\State \Return $\text{Chains}$	  
	\end{algorithmic}	
\end{algorithm}

  \emph{Algorithm 2}: Note that Algorithm \ref{alg4} is  simple and fast, but may not find good chains to improve the MSE in (\ref{baseline}). Therefore, we are interested in finding  good chains and  the corresponding region of $(\boldsymbol{\Delta_c},\boldsymbol{\Delta_s})$ such that the upper bound of MSE is smaller  than (\ref{baseline}).  
  We illustrate our idea with  a special case where the  length of each chain is less than 2.

  
  For Client $i$, consider $y_{t}\rightarrow x_{t}\rightarrow x_{i}$ and $y_{i}\rightarrow x_{i}$ with $t<i$. By  Remark \ref{al2} (described later in Section \ref{secUpper}), we select the chain as follows: If $(\Delta_t,\Delta_{ti},\Delta_i)\in \mathcal{R}_2$ is in the  region $\mathcal{R}_2$ defined in \eqref{Region2},  then we use the chain $\mathcal{L}_i$ as $y_{t}\rightarrow x_{t}\rightarrow x_{i}$, otherwise   we  use the chain  $y_{i}\rightarrow x_{i}$.   Now  we look for good chains according to  $\mathcal{R}_2$. Without loss of generality, let $\Delta_1\leq\cdots\leq\Delta_n$.  Starting from Client $1$, firstly, generate a chain of length of $1$ for Client $1$, and then traverse the remaining clients to verify whether the corresponding distances are in $\mathcal{R}_2$. If the distances are in $\mathcal{R}_2$, then construct a chain of length of $2$. For the remaining clients  whose chains are empty, renumber them and repeat the above process until every client has a nonempty chain. 
  We formally describe the method in Algorithm~\ref{alg5}. 

   \begin{algorithm}
   	\caption{Selection of chains for special case}\label{alg5}
   	\hspace*{0.02in}{\bf Input:}
   	Input $(\boldsymbol{\Delta_s},\boldsymbol{\Delta_c}), \mathcal{R}_2$\\
   	\hspace*{0.02in}{\bf Output:}
   	$Chains$
   	\begin{algorithmic}[1]
   		\State ${Chains}\leftarrow \emptyset$
   		\State $C\leftarrow\text{ List}[1,2\cdots,n]$
   		\While{$C$ is not empty}
   		\State $Node\leftarrow \emptyset$
   		\State $Node\leftarrow Node\cup \{C[0]\}$
   		\State Generate $y_{C[0]}\rightarrow x_{C[0]}$ for Client $C[0]$
   		\State $Chains\leftarrow Chains\cup\{y_{C[0]}\rightarrow x_{C[0]}\}$
   		\For{$i>0$}
   		\If{$(\Delta_{C[0]},\Delta_{C[0]C[i]},\Delta_{C[i]})\in \mathcal{R}_2$}
   		\State Generate $y_{C[0]}\rightarrow x_{C[0]}\rightarrow x_{C[i]}$ for Client $C[i]$
   		\State $Chains\leftarrow Chains\cup\{y_{C[0]}\!\rightarrow x_{C[0]}\rightarrow\! x_{C[i]}\}$
   		\State $Node\leftarrow Node\cup \{C[i]\}$
   		\EndIf
   		\EndFor  
   		\State Deleta clients in $Node$ from $C$  
   		\EndWhile
   		\State \Return $\text{Chains}$	  
   	\end{algorithmic}	
   \end{algorithm}

\subsection{New Upper Bound of MSE}\label{secUpper}
Define the   following quantities:
\begin{IEEEeqnarray}{rCl}\label{eqDefab}
	\alpha_i(Q)\triangleq \sup_{\boldsymbol{x},\boldsymbol{y}\ \text{satisfy}\ (\ref{eqDist})}{\mathbb{E}[\|\hat{x}_{Q,i}-{x}_i||_2^2]},\label{msey}\\
	\beta_i(Q)\triangleq \sup_{\boldsymbol{x},\boldsymbol{y}\ \text{satisfy}\ (\ref{eqDist})}{||\mathbb{E}[\hat{x}_{Q,i}-x_i]\|_2^2}.\label{biasy}
\end{IEEEeqnarray}

Given a chain $\mathcal{L}_i$ and a specific parameters assignment, the following lemma gives a recursive inequality about the upper bound of the error  when using the quantizer $Q^{\mathcal{L}_i}_{\textnormal{Pro}}$.

\begin{lemma}\label{main1}
	For a given chain $\mathcal{L}_i:y_{i_1}\rightarrow x_{i_1}\rightarrow x_{i_2}\cdots\rightarrow x_{i_l}$, when using the quantizer $Q^{\mathcal{L}_i}_{\textnormal{Pro}}$  with the following parameters: $k\geq 4$, $\Delta_{i_1}'=\sqrt{6(\Delta_{i_1}^2/d)\ln \sqrt{n}}$, $\Delta_{i_si_{s+1}}'=\sqrt{6(\Delta_{i_si_{s+1}}^2/d)\ln \sqrt{n}}$, where $s\in[l-1]$, $\mu d=\lfloor r/ \log k\rfloor$ and   $\epsilon=2\Delta_{i_1}'/(k-2)+\sum_{s={1}}^{l-1}2\Delta_{i_si_{s+1}}'/(k-2)$,  we have 
	\begin{equation*}
	\begin{aligned}
	\alpha_{i_t}(Q^{\mathcal{L}_{i_t}}_{\textnormal{Pro}} )\!&\leq\!\frac{24t(\Delta_{i_1}^2\!+\!\sum_{s=1}^{t-1}\Delta_{i_si_{(s+1)}}^2)\!\ln\!\sqrt{n}}{\mu(k\!-\!2)^2}\\
	&\quad+c_{t}(n)\frac{\Delta_{i_1}^2\!+\!\sum_{s=1}^{t-1}\Delta_{i_si_{(s+1)}}^2\!+\!\Delta_{i_{(t\!-\!1)}i_t}^2}{\mu n}\!\\
	&\quad+\frac{\!3\alpha_{i_{t-1}}(Q^{\mathcal{L}_{i_{t\!-\!1}}}_{\textnormal{Pro}})+\Delta_{i_t}^2}{\mu},
	\end{aligned}
	\end{equation*}
	\begin{equation*}
	\begin{aligned}
	\beta_{i_t}(Q^{\mathcal{L}_i}_{\textnormal{Pro}})&\leq c_t(n)\frac{\Delta_{i_1}^2\!+\!\sum_{s=1}^{t-1}\Delta_{i_si_{(s+1)}}^2+\Delta_{i_{(t\!-\!1)}i_t}^2}{n}\\
	&\quad+3\alpha_{i_{t-1}}(Q^{\mathcal{L}_{i_t}}_{\textnormal{Pro}}),
	\end{aligned}
	\end{equation*}
	and \begin{equation*}
	\alpha_{i_1}(Q^{\mathcal{L}_{i_1}}_{\textnormal{Pro}})\leq\frac{24\Delta_{i_1}^2\ln\sqrt{n}}{\mu(k-2)^2}+  154\frac{\Delta_{i_1}^2}{\mu n}+\frac{\Delta_{i_1}^2}{\mu},
	\end{equation*}
	\begin{equation*}
	\beta_{i_1}(Q^{\mathcal{L}_{i_1}}_{\textnormal{Pro}})\leq 154\frac{\Delta_{i_1}^2}{n},
	\end{equation*}
	for all $t\in[l]$ and   $c_{t}(n)\!\triangleq\!\max\{\frac{576t^2}{e},3n\!+\!\frac{36}{e^{2/3}}\}$.
\end{lemma}
\begin{proof}
	See the proof in Appendix \ref{pro4}.
\end{proof}

By properly scaling and choosing an appropriate $k$, we obtain a more concise form in the following corollary, whose proof   is given in Appendix \ref{c1}.

\begin{corollary}\label{simple}
	If setting $\log k=\lceil \log(2+\sqrt{12\ln n})\rceil $, we have 
	\begin{equation*}
	\alpha_{i_l}(Q^{\mathcal{L}_{i_l}}_{\textnormal{Pro}})\leq\frac{24D_{i_l}^i\ln\sqrt{n}}{\mu(k-2)^2}+154\frac{D_{i_l}^i}{\mu n}+\frac{\Delta_{i_l}^2}{\mu},
	\end{equation*}
	\begin{equation*}
	\beta_{i_l}(Q^{\mathcal{L}_{i_l}}_{\textnormal{Pro}})\leq 154\frac{D_{i_l}^i}{n},
	\end{equation*}
	 and  $D_{i_l}^i$ satisfies that for $l=1$, $D_{i_1}^i=\Delta_1^2$ and for $l>1$, \begin{IEEEeqnarray}{rCl}\label{imp}
	D_{i_l}^i=&&\max\Big\{l(\Delta_{i_1}^2\!+\!\sum_{s=1}^{l-1}\Delta_{i_si_{(s+1)}}^2), \frac{c_{l}(n)}{154}(\Delta_{i_1}^2\\
	\nonumber && +\!\sum_{s=1}^{l-1}\Delta_{i_si_{(s+1)}}^2+\Delta_{i_{(l\!-\!1)}i_l}^2)\!+\!\frac{3nD_{i_{l-1}}^i}{154}\Big\}\!+\!3D_{i_{l-1}}^i. 
\end{IEEEeqnarray}
\end{corollary}


 \begin{theorem}\label{final}
	For some fixed $(\boldsymbol{\Delta_s},\boldsymbol{\Delta_c})$,  and $d\geq r\geq 2\lceil \log(2+\sqrt{12\ln n})\rceil$, and $\mu d=\lfloor\frac{r}{\log k}\rfloor$, the MSE is upper bounded by
	\begin{equation}\label{OurBound}
		\begin{aligned}
		\mathcal{E}_{\pi^*}(\boldsymbol{x},\boldsymbol{y})
		&\leq(79\log k+26)\sum_{i=1}^{n}{\frac{d\Delta_i^2}{n^2r}}+B\sum_{i=1}^{n}{\frac{d(D^i_{i_l} -\Delta_i^2)}{n^2r}},\\
		\end{aligned}
	\end{equation}
	for all sets of chains $\{\mathcal{L}_i\}_{i=1}^n$ and 
	\begin{IEEEeqnarray*}{rCl}
	B=\left\{\begin{array}{llr}
	79 \log k+26, &~\text{if}~\sum_{i=1}^{n}{(D^i_{i_l}-\Delta_i^2)}\geq0\\
	 \frac{\log k}{8}.&~\text{otherwise}	
	\end{array}\right.,
\end{IEEEeqnarray*} 
 where    $\log k=\lceil \log(2+\sqrt{12\ln n})\rceil$ and $D_{i_l}^i$ is given in \eqref{imp}.
 \end{theorem}

\begin{Remark}\label{re}
	We improve the upper bound of MSE  in (\ref{baseline}) when $(\boldsymbol{\Delta_c},\boldsymbol{\Delta_s})$ are in the region $\mathcal{R}_{\mathcal{L}}=\{(\boldsymbol{\Delta_c},\boldsymbol{\Delta_s}):\sum_{i=1}^{n}D_{i_l}^i<\sum_{i=1}^{n}\Delta_i^2\} $. 
In the  region   $\mathcal{R}_{\mathcal{L}}$ our new upper bound in \eqref{OurBound} is $1\!-\!\frac{\log k}{8(79\log k\!+\!26)}(1-\frac{\sum_{i=1}^{n}D_{i_l}^i}{\sum_{i=1}^{n}\Delta_i^2})$ times of that in (\ref{baseline}). 
\end{Remark}

\begin{Remark}
For each permutation $\Pi$ on [n], since each Client $\Pi_i$ can choose $\mathcal{L}_{\Pi_i}$ from $i$ candidate chains, there are $n!$ different assignments $\{\mathcal{L}_{\Pi_i}:i\in[n]$\}.  Thus,  the number of all strategies will not exceed $(n!)^2$. Denote the chain corresponding to each strategy as $\mathcal{L}^i, i\in[(n!)^2]$. From Theorem \ref{final} and Remark \ref{re}, we obtain that when $(\boldsymbol{\Delta_c},\boldsymbol{\Delta_s})\in \bigcup_{i\in[(n!)^2]}\mathcal{R}_{\mathcal{L}^i}$, our upper bound is tighter than that in \eqref{baseline}. 
\end{Remark}
\begin{Remark}\label{al2}
	If Client $i$ uses the chain $y_{t}\rightarrow x_{t}\rightarrow x_{i}$, by (\ref{imp}), we have 
	$D_{i_2}^i=\max\{2(\Delta_t^2+\Delta_{ti}^2), \frac{c_2(n)}{154}(\Delta_t^2+2\Delta_{ti}^2)+\frac{3n\Delta_t^2}{154}\}+3\Delta_t^2$,
	where $c_{2}(n)=\max\{\frac{2304}{e},3n+\frac{36}{e^{2/3}}\}$.
	If Client $i$ uses the chain $y_{i}\rightarrow x_{i}$, by (\ref{imp}),  we have $D_{i_1}^i=\Delta_i^2$. Let 
		\begin{IEEEeqnarray}{rCl}\label{Region2}
  &&\mathcal{R}_2\triangleq\Big\{(\Delta_t,\Delta_{ti},\Delta_i):\max\big\{2(\Delta_t^2+\Delta_{ti}^2)+3\Delta_t^2,\nonumber\\
  &&\quad \quad c_{2}(n)(\Delta_t^2\!+\!2\Delta_{ti}^2)\!+\!\frac{3n\Delta^2_t}{154}\!+\!3\Delta_t^2\big\}<\Delta_i^2\Big\}.
\end{IEEEeqnarray} 
	
Then,  by Theorem \ref{final}, if there exists   some $i,t\in[n]$ such that    $(\Delta_i,\Delta_t,\Delta_{ti})\in \mathcal{R}_2$, then  our estimator $Q_\textnormal{pro}$   improves the  Wyner-Ziv estimator  $Q_\textnormal{WZ}$ proposed in \cite{20}.

\end{Remark}

\begin{Remark}
From    \eqref{Region2}, we observe that to choose a chain whose length is larger than 2,  there must exist at least one pair   $(\Delta_t,\Delta_i)$  such that $5\Delta^2_t< \Delta^2_i$, for some $t,i\in[n]$. Otherwise, our estimator turns to be the Wyner-Ziv estimator in \cite{20}.  The condition $5\Delta^2_t< \Delta^2_i$ seems  a stringent assumption at the first glance. In fact,  since our   quantizer is for   specific vectors $\boldsymbol{x}$ and  $\boldsymbol{y}$, and the Euclidean distance   $\|x_i-y_i\|\leq \Delta_i$ and $\|x_j-y_j\|\leq \Delta_j$, for $i,j\in[n]$ can vary greatly. Also,  one can spend additional bits on better estimating   $\Delta_t$ and $\Delta_{ti}$   such that they  are smaller enough to satisfy \eqref{Region2}, and this additional bits cost on estimating  $\Delta_t$ and $\Delta_{ti}$ can be omitted when $n$  is relatively large.
\end{Remark}

\section{Proof of Theorem \ref{final}}

Now we first introduce a lemma and then derive  an upper bound of MSE for any $r$-bit  quantizer. 
\begin{lemma}{(see \cite{20})}\label{upbound}
	For $ \boldsymbol{x}$ and $\boldsymbol{y}$ satisfying (\ref{eqDist}), and an $r$-bit quantizer $Q$ using independent randomness for different $i\in[n]$, the estimate $\hat{\bar{x}}$ in (\ref{esti}) and the sample mean $\bar{x}$ satisfies
	\begin{equation}\label{bound}
	\mathbb{E}[||\hat{\bar{x}}-\bar{x}||_2^2]\leq \sum_{i=1}^{n}{\frac{\alpha_{i}(Q)}{n^2}}+\sum_{i=1}^{n}{\frac{\beta_{i}(Q)}{n}}.
	\end{equation}
\end{lemma}
	
	By Lemma \ref{upbound}, we have
	\begin{equation}\label{bounds}
	\begin{aligned}
	\mathcal{E}_{\pi^*}&(\boldsymbol{x},\boldsymbol{y})\\
	&\leq\sum_{i=1}^{n}{\frac{\alpha_{i}(Q^{\mathcal{L}_i}_{\textnormal{Pro}})}{n^2}}+\sum_{i=1}^{n}{\frac{\beta_{i}(Q^{\mathcal{L}_i}_{\textnormal{Pro}})}{n}}\\
	&\overset{(a)}{\leq}\frac{r}{\mu d}(\frac{24\ln\sqrt{n}}{(k-2)^2}\!+\!\frac{154}{n}\!+\!154\mu)\sum_{i=1}^{n}{\frac{dD_{i_l}^i}{n^2r}}\!+\!\frac{r}{\mu d}\sum_{i=1}^{n}{\frac{d\Delta_i^2}{n^2r}}\\
	&{=}\frac{r}{\mu d}(\frac{24\ln\sqrt{n}}{(k-2)^2}\!+\!\frac{154}{n}\!+\!1\!+\!154\mu)\sum_{i=1}^{n}{\frac{d\Delta_i^2}{n^2r}}\\
	&\quad+\frac{r}{\mu d}(\frac{24\ln\sqrt{n}}{(k-2)^2}\!+\!\frac{154}{n}\!+\!154\mu)\sum_{i=1}^{n}{\frac{d(D_{i_l}^i-\Delta_i^2)}{n^2r}}\\
	&\overset{(b)}{\leq}(79\lceil \log(2+\sqrt{12\ln n})\rceil+26)\sum_{i=1}^{n}{\frac{d\Delta_i^2}{n^2r}}\\
	&\quad+\underbrace{\frac{r}{\mu d}(\frac{24\ln\sqrt{n}}{(k-2)^2}\!+\!\frac{154}{n}\!+\!154\mu)}_{B_{1}}\sum_{i=1}^{n}{\frac{d(D_{i_l}^i-\Delta_i^2)}{n^2r}},\\
	\end{aligned}	
	\end{equation}
	where (a) follows by Corollary \ref{simple} and $x_{i_l}=x_{i}$, and  (b) follows by the  inequality $\frac{r}{\mu d}(\frac{24\ln\sqrt{n}}{(k-2)^2}\!+\!\frac{154}{n}\!+\!1\!+\!154\mu) 
		 \leq(79\lceil \log(2+\sqrt{12\ln n})\rceil+26)$ given    in \cite{20}.   By this inequality, we can obtain an upper bound of $B_1$ as 
		$ B_1 \leq\frac{r}{\mu d}(\frac{24\ln\sqrt{n}}{(k-2)^2}\!+\!\frac{154}{n}\!+\!1\!+\!154\mu) 
		 \leq(79\lceil \log(2+\sqrt{12\ln n})\rceil+26)$, 
		 and an lower bound    \begin{equation*}
		\begin{aligned}
	 B_1&\geq\frac{r}{\mu d}(\frac{24\ln\sqrt{n}}{(k-2)^2})\overset{(a)}{\geq}\frac{r}{\mu d}(\frac{24\ln\sqrt{n}}{4(1+\sqrt{12\ln n})^2})\\
	 &\overset{(b)}{\geq}\frac{r}{8\mu d}\overset{(c)}{\geq} \frac{\lceil \log(2+\sqrt{12\ln n})\rceil}{8},
	 \end{aligned}
	\end{equation*}
where   (a) holds by $\log k\leq \log(2+\sqrt{12\ln n})+1$,   (b) is because   $\frac{24\ln\sqrt{n}}{4(1+\sqrt{12\ln n})^2}\geq\frac{24\ln\sqrt{2}}{4(1+\sqrt{12\ln 2})^2}\geq1/8$ and (c) holds by $\mu d=\lfloor\frac{r}{\log k}\rfloor\leq\frac{r}{\log k}$.
\section{Conclusion}
In this paper, we studied the distributed mean estimation with limited communication. Inspired by Wyner-Ziv and Slepian-Wolf coding, we proposed new estimator by  exploiting the correlation  between clients' data. In the future work, we aim to find a more efficient estimator   and apply it to more generalized distributed optimization framework.


\bibliographystyle{IEEEtran.bst}
\bibliography{1.bib}
\appendix

\subsection{Proof of Lemma \ref{main1}}\label{pro4}

Our quantizer is based on $Q_{\textnormal{M},R}$, similar to $Q_{\textnormal{WZ}}$. We first introduce the following lemma, whose proof is similar to that in \cite{20}. With a slight abuse of notation, we write $Q^{\mathcal{L}_{i_l}}_\textnormal{Pro}$ as $Q_\textnormal{Pro}$.


\begin{lemma}\label{lossyq}
	Fix $\Delta_i >0$. Then, for $\mu d\in[d]$, we have \begin{equation*}\label{uboundr}
	\alpha_{i_l}(Q_{\textnormal{Pro}})\leq\frac{\alpha_{i_l}(Q_{\textnormal{M},R})}{\mu}+\frac{\Delta_{i_l}^2}{\mu};
	\end{equation*}
	\begin{equation*}\label{uboundr2}
	\beta_{i_l}(Q_{\textnormal{Pro}})=\beta_{i_l}(Q_{\textnormal{M},R}).
	\end{equation*}
\end{lemma}
\begin{proof}
	\begin{equation*}
	\begin{aligned}
	\mathbb{E}&[||\hat{x}_{Q_{\textnormal{Pro}},{i_l}}- x_{i_l}||_2^2]\\
	&
	=\sum_{j\in [d]}\mathbb{E}\Big[(\frac{1}{\mu}(R\hat{x}_{Q_{\textnormal{M},R},i_l}(j)-Ry_{i_l}(j))\mathbbm{1}_{\{j\in S\}}\\ &\quad-(Rx_{i_l}(j)-Ry_{i_l}(j)))^2\Big]\\
	&
	=\sum_{j\in [d]}\mathbb{E}[(\frac{1}{\mu}(R\hat{x}_{Q_{\textnormal{M},R},i_l}\!-\!Rx_{i_l}(j)))^2\mathbbm{1}_{\{j\in S\}}]\\&
	\quad\!+\!\sum_{j\in [d]}\mathbb{E}\Big[(\frac{1}{\mu}(Rx_{i_l}(j)\!-\!Ry_{i_l}(j))\mathbbm{1}_{\{j\in S\}}\nonumber\\
	&\quad\quad\quad\quad\quad\quad \!-\!(Rx_{i_l}(j)\!-\!Ry_{i_l}(j)))^2\Big]
	\\
	&
	=\frac{1}{\mu}\sum_{j\in [d]}\mathbb{E}[(R\hat{x}_{Q_{\textnormal{M},R},i_l}-Rx_{i_l}(j))^2]\\&\quad+
	\sum_{j\in [d]}\mathbb{E}[(Rx_{i_l}(j)-Ry_{i_l}(j))^2]\cdot\mathbb{E}[(\frac{1}{\mu}\mathbbm{1}_{\{j\in S\}}-1)^2]
	\\
	&
	=\frac{1}{\mu}\sum_{j\in [d]}\mathbb{E}[(R\hat{x}_{Q_{\textnormal{M},R},i_l}-Rx_{i_l}(j))^2]\\&\quad+
	\sum_{j\in [d]}\mathbb{E}[(Rx_{i_l}(j)-Ry_{i_l}(j))^2]\cdot\frac{1-\mu}{\mu}
	\\
	&\leq\frac{\alpha_{i_l}(Q_{\textnormal{M},R})}{\mu}+\frac{\Delta_{i_l}^2}{\mu},
	\end{aligned}
	\end{equation*}
	where we use the independence of $S$ and $R$ in the third identity and use the fact that $R$ is unitary in the final step.
	
	Since 
	\begin{equation*}
	\begin{aligned}
	||&\mathbb{E}[\hat{x}_{Q_{\textnormal{Pro}},i_l}]- x_{i_l})||_2^2\\
	&
	=||\sum_{j\in [d]}\mathbb{E}[(\frac{1}{\mu}(R\hat{x}_{Q_{\textnormal{M},R},i_l}(j)-Ry_{i_l}(j)))\mathbbm{1}_{\{j\in S\}}\\
	&\quad-(Rx_{i_l}(j)-Ry_{i_l}(j))]\boldsymbol{e}_j||_2^2
	\\
	&=||\mathbb{E}[\hat{x}_{Q_{\textnormal{M},R}, i_l}]- x_{i_l})||_2^2,
	\end{aligned}
	\end{equation*}
	where we use the independence of $S$ and $Q_{\textnormal{M},R}$ in the last identity, we have
	\begin{equation*}
	\beta_{i_l}(Q_{\textnormal{Pro}})=\beta_{i_l}(Q_{\textnormal{M},R}).
	\end{equation*}
\end{proof}
The following lemma is given in \cite{20}, which shows that $Q_\textnormal{M}$ is unbiased under certain conditions, and the error will not exceed $\epsilon$.
\begin{lemma}{(see \cite{20})}\label{cqpe}
	Consider $Q_\textnormal{M}$ described in \ref{cq} with parameter $\epsilon$ set to satisfy
	\begin{equation}\label{condk}
	k\epsilon\geq2(\epsilon+\Delta').
	\end{equation}
	Then, for every $x,h\in\mathbb{R}$ such that $|x-h|\leq\Delta'$, the output $Q_\textnormal{M}(x)$ satisfies
	\begin{equation*}
	\begin{array}{c}
	\mathbb{E}[Q_\textnormal{M}(x)]=x\\
	|x-Q_\textnormal{M}(x)|<\epsilon.
	\end{array}
	\end{equation*}
\end{lemma}

Recall from Section \ref{chain} that for a chain $y_{i_1}\rightarrow x_{i_1}\rightarrow x_{i_2}\cdots\rightarrow x_{i_l}$, the server estimates $ x_{i_l}$ by using $\hat{x}_{Q^{\mathcal{L}_i}_{\textnormal{Pro},i},i_{l-1}}$ as $h$ and $\Delta_{i_1}'+\sum_{s=1}^{l-1}\Delta_{i_si_{(s+1)}}'$ as the parameter $\Delta'$ in $Q_{\textnormal{M},R}$. By Lemma \ref{lossyq}, we consider the quantizer $Q_{\textnormal{M},R}$.
\begin{lemma}\label{chainde}
	If $R$ given in (\ref{ger}) satisfies that for $j\in[d]$ and $t\in[l]$, $|R x_{i_1}(j)-Ry_{i_1}(j)|\leq\Delta'_{i_1}$, $|R x_{i_s}(j)-R x_{{i_{s+1}}}(j)|\leq\Delta_{i_si_{(s+1)}}',\ s\in[l-1]$, then, for $k\geq4$, we have 
	\begin{equation*}
	|R x_{i_t}(j)-R\hat{x}_{Q_{\textnormal{M},R},{i_{t-1}}}(j)|\leq\Delta_{i_1}'+\sum_{s=1}^{t-1}\Delta_{i_si_{(s+1)}}',\ t>1,
	\end{equation*}
	and if $t=1$, we denote $\hat{x}_{Q_{\textnormal{M},R},0}$ by $y_{i_1}$. 
\end{lemma}
\begin{proof}
	We use induction for $l$. If $l=2$, since $|R x_{i_1}(j)-Ry_{i_1}(j)|\leq\Delta'_{i_1}$, by Lemma \ref{cqpe}, we have 
	\begin{equation*}
	|R x_{i_1}(j)-R\hat{x}_{Q_{\textnormal{M},R},{i_1}}(j)|\leq\frac{2\Delta'_{i_1}}{k-2},
	\end{equation*}
	where we set $\epsilon=\frac{2\Delta'_{i_1}}{k-2}$. Then, for $k\geq4$
	\begin{equation*}
	\begin{aligned}
	|R& x_{i_2}(j)-R\hat{x}_{Q_{\textnormal{M},R},i_1}(j)|\\
	&\leq|R x_{i_2}(j)-R x_{i_1}(j)+R x_{i_1}(j)-R\hat{x}_{Q_{\textnormal{M},R},i_1}(j)|\\
	&\leq|R x_{i_2}(j)-R x_{i_1}(j)|+|R x_{i_1}(j)-R\hat{x}_{Q_{\textnormal{M},R},i_1}(j)|\\
	&\leq\Delta_{i_1i_2}'+\frac{2\Delta'_{i_1}}{k-2}\\
	&\leq\Delta_{i_1i_2}'+\Delta'_{i_1}.
	\end{aligned}
	\end{equation*}
	Suppose
	\begin{equation*}
	|R x_{i_{t-1}}(j)-R\hat{x}_{Q_{\textnormal{M},R},i_{t-2}}(j)|\leq\Delta_{i_1}'+\sum_{s=1}^{t-2}\Delta_{i_si_{(s+1)}}'.
	\end{equation*}
	By Lemma \ref{cqpe}, we have
	\begin{equation*}
	|R x_{i_{t-1}}(j)-R\hat{x}_{Q_{\textnormal{M},R},i_{t-1}}(j)|\leq\frac{2(\Delta'_{i_1}+\sum_{s=1}^{t-2}\Delta'_{i_si_{(s+1)}})}{k-2}.
	\end{equation*}
	So, for $t$ and $k\geq4$, we have 
	\begin{equation*}
	\begin{aligned}
	|R& x_{i_t}(j)-R\hat{x}_{Q_{\textnormal{M},R},i_{t-1}}(j)|\\
	&\leq|R x_{i_t}(j)-R x_{i_{t-1}}(j)+R x_{i_{t-1}}(j)-R\hat{x}_{Q_{\textnormal{M},R},i_{t-1}}(j)|\\
	&\leq|R x_{i_t}(j)-R x_{i_{t-1}}(j)|+|R x_{i_{t-1}}(j)-R\hat{x}_{Q_{\textnormal{M},R},i_{t-1}}(j)|\\
	&\leq\Delta_{i_{(t-1)}i_t}'+\frac{2(\Delta_{i_1}'+\sum_{s=1}^{t-2}\Delta_{i_si_{(s+1)}}')}{k-2}\\
	&\leq\Delta_{1}'+\sum_{s=1}^{t-1}\Delta'_{i_si_{(s+1)}}.
	\end{aligned}
	\end{equation*}
\end{proof}
For convenience, let $\mathcal{A}_1$ denote event $\{R:|R x_{i_1}(j)-Ry_{i_1}(j)|\leq\Delta'_{i_1}\}$ and $\mathcal{A}_{s}$ denote event $\{R:|R x_{i_{s-1}}(j)-R x_{i_s}(j)|\leq\Delta'_{i_{(s-1)}i_s}\}$. So, by Lemma \ref{chainde}, we have that if $R\in \bigcap_{s\in[t]}\mathcal{A}_s$, then \begin{equation*}
|R x_{i_t}(j)-R\hat{x}_{Q_{\textnormal{M},R},i_{t-1}}(j)|\leq\Delta_{i_1}'+\sum_{s=1}^{t-1}\Delta'_{i_si_{(s+1)}}.
\end{equation*}
Thus, \begin{equation}\label{estiA}
\begin{aligned}
|R x_{i_t}(j)-R\hat{x}_{Q_{\textnormal{M},R},i_t}(j)|&\leq\epsilon_{i_1}+\sum_{s=1}^{t-1}\epsilon_{i_si_{(s+1)}}\\
&=\frac{2(\Delta'_{i_1}+\sum_{s=1}^{t-1}\Delta'_{i_si_{(s+1)}})}{k-2}.
\end{aligned}
\end{equation}
where we set $\epsilon_{i_1}=\frac{2\Delta_{i_1}'}{k-2}$ and $\epsilon_{i_si_{(s+1)}}=\frac{2\Delta'_{i_si_{(s+1)}}}{k-2}$.

For the random matrix $R$ given in (\ref{ger}), for every $\boldsymbol{z}\in \mathbb{R}^d$, the random variables $R\boldsymbol{z}(i),\ i\in[d]$, are sub-Gaussian with variance parameter $||\boldsymbol{z}||_2^2/d$. Furthermore, we need the following bound.
\begin{lemma}[see \cite{20}]\label{sG}
	For a sub-Gaussian random $Z$ with variance factor $\sigma^2$ and every $t\geq0$, we have \begin{equation*}
	\mathbb{E}[Z^2\mathbbm{1}_{\{|Z|>t\}}]\leq 2(2\sigma^2+t^2)e^{-t^2/2\sigma^2}.
	\end{equation*}
\end{lemma}

We now handle the $\alpha_{i_t}(Q_{\textnormal{M},R})$ and $\beta_{i_t}(Q_{\textnormal{M},R})$ separately below. 

Firstly, we consider $t>1$. Since $R$ is a unitary transform, we have 
\begin{equation}\label{eq1}
\begin{aligned}
\mathbb{E}&[||\hat{x}_{Q_{\textnormal{M},R},i_t}- x_{i_t}||_2^2]\\
&=\mathbb{E}[||R\hat{x}_{Q_{\textnormal{M},R},i_t}-R x_{i_t}||_2^2]\\
&=\sum_{j=1}^{d}\mathbb{E}[(R\hat{x}_{Q_{\textnormal{M},R},i_t}(j)-R x_{i_t}(j))^2]\\
&=
\sum_{j=1}^{d}\mathbb{E}[(R\hat{x}_{Q_{\textnormal{M},R},i_t}(j)-R x_{i_t}(j))^2\mathbbm{1}_{\bigcap_{s\in[t]}\mathcal{A}_s}]\\
&\quad+
\sum_{j=1}^{d}\mathbb{E}[(R\hat{x}_{Q_{\textnormal{M},R},i_t}(j)-R x_{i_t}(j))^2\mathbbm{1}_{\bigcup_{s\in[t]}\mathcal{A}_s^c}].
\end{aligned}
\end{equation}

We consider the first term. By (\ref{estiA}), we have 
\begin{equation}\label{eq2}
\begin{aligned}
\sum_{j=1}^{d}&\mathbb{E}[(R\hat{x}_{Q_{\textnormal{M},R},i_t}(j)-R x_{i_t}(j))^2\mathbbm{1}_{\bigcap_{s\in[t]}\mathcal{A}_s}]\\
&\leq d(\epsilon_{i_1}+\sum_{s=1}^{t-1}\epsilon_{i_si_{(s+1)}})^2\\
&\leq dt(\epsilon_{i_1}^2+\sum_{s=1}^{t-1}\epsilon_{i_si_{(s+1)}}^2),
\end{aligned}
\end{equation}
where in the final steps we use the fact that $(a_1+\cdots+a_n)^2\leq n(a_1^2+\cdots+a_n^2)$.

For the second term on (\ref{eq1}), we get 
\begin{equation}\label{eq3}
\begin{aligned}
\sum_{j=1}^{d}&\mathbb{E}[(R\hat{x}_{Q_{\textnormal{M},R},i_t}(j)-R x_{i_t}(j))^2\mathbbm{1}_{\bigcup_{s\in[t]}\mathcal{A}_s^c}]\\
&
\leq 3\sum_{j=1}^{d}\Big[\mathbb{E}[(R\hat{x}_{Q_{\textnormal{M},R},i_t}(j)-R\hat{x}_{Q_{\textnormal{M},R},i_{t-1}}(j))^2\mathbbm{1}_{\bigcup_{s\in[t]}\mathcal{A}_s^c}]\\
&\quad+
\mathbb{E}[(R\hat{x}_{Q_{\textnormal{M},R},i_{t-1}}(j)-R x_{i_{t-1}}(j))^2\mathbbm{1}_{\bigcup_{s\in[t]}\mathcal{A}_s^c}]\\
&\quad+
\mathbb{E}[(R x_{i_{t-1}}(j)-R x_{i_t}(j))^2\mathbbm{1}_{\bigcup_{s\in[t]}\mathcal{A}_s^c}]\Big]
\\
&
\leq 3k^2(\epsilon_{i_1}+\sum_{s=1}^{t-1}\epsilon_{i_si_{(s+1)}})^2\sum_{j=1}^{d}P(\bigcup_{s\in[t]}\mathcal{A}_s^c)\\
&\quad+
3\sum_{j=1}^{d}\Big[\mathbb{E}[(R\hat{x}_{Q_{\textnormal{M},R},i_{t-1}}(j)-R x_{i_{t-1}}(j))^2\mathbbm{1}_{\bigcup_{s\in[t]}\mathcal{A}_s^c}]\\
&\quad+
\mathbb{E}[(R x_{i_{t-1}}(j)-R x_{i_t}(j))^2\mathbbm{1}_{\bigcup_{s\in[t]}\mathcal{A}_s^c}]\Big]
\\
&
\leq 3k^2(\epsilon_{i_1}+\sum_{s=1}^{t-1}\epsilon_{i_si_{(s+1)}})^2\sum_{j=1}^{d}P(\bigcup_{s\in[t]}\mathcal{A}_s^c)\\
&\quad\!+\!
3\sum_{j=1}^{d}\mathbb{E}[(R x_{i_{t-1}}(j)\!-\!R x_{i_t}(j))^2\mathbbm{1}_{\bigcup_{s\in[t]}\mathcal{A}_s^c}]\!+\!3\alpha_{i_{t-1}}(Q_{\textnormal{M},R})
\\
&
\leq 3k^2(\epsilon_{i_1}+\sum_{s=1}^{t-1}\epsilon_{i_si_{(s+1)}})^2\sum_{j=1}^{d}P(\bigcup_{s\in[t]}\mathcal{A}_s^c)\\
&\quad+
3\sum_{j=1}^{d}\mathbb{E}[(R x_{i_{t-1}}(j)-R x_{i_t}(j))^2\mathbbm{1}_{\bigcup_{s\in[t-1]}\mathcal{A}_s^c}]\\
&\quad+
3\sum_{j=1}^{d}\mathbb{E}[(R x_{i_{t-1}}(j)-R x_{i_t}(j))^2\mathbbm{1}_{\mathcal{A}_l^c}]+
3\alpha_{i_{t-1}}(Q_{\textnormal{M},R})
\\
&
\leq 3k^2(\epsilon_{i_1}+\sum_{s=1}^{t-1}\epsilon_{i_si_{(s+1)}})^2\sum_{j=1}^{d}P(\bigcup_{s\in[t]}\mathcal{A}_s^c)\\
&\quad\!+\!
3\Delta_{i_{(t-1)}i_t}^2\!+\!
3\sum_{j=1}^{d}\mathbb{E}[(R x_{i_{t-1}}(j)\!-\!R x_{i_t}(j))^2\mathbbm{1}_{\mathcal{A}_t^c}]
\\
&\quad\!+\!
3\alpha_{i_{t-1}}(Q_{\textnormal{M},R}).
\end{aligned}
\end{equation}

Since \begin{equation}\label{seq1}
\begin{aligned}
&P(\bigcup_{s\in[t]}\mathcal{A}_s^c)\\
&\leq\sum_{s=1}^{t}P(\mathcal{A}_s^c)\\
&\leq 2e^{-d\Delta_{i_1}^{'2}/2\Delta_{i_1}^2}+2\sum_{s=2}^{t}e^{-d\Delta_{i_{(s-1)}i_s}^{'2}/2\Delta_{i_{(s-1)}i_s}^2}\\
&=2t(\sqrt{n})^{-3},
\end{aligned}
\end{equation}
where in the final step we use $\Delta_{i_1}'=\sqrt{6(\Delta_{i_1}^2/d)\ln \sqrt{n}}$ and $\Delta_{i_{(s-1)}i_s}'=\sqrt{6(\Delta_{i_{(s-1)}i_s}^2/d)\ln \sqrt{n}}$, we get
\begin{equation}\label{seq2}
\begin{aligned}
\sum_{j=1}^{d}&\mathbb{E}[(R\hat{x}_{Q_{\textnormal{M},R},i_t}(j)-R x_{i_t}(j))^2\mathbbm{1}_{\bigcup_{s\in[t]}\mathcal{A}_s^c}]\\
& 
\leq 6tk^2(\epsilon_{i_1}+\sum_{s=1}^{t-1}\epsilon_{i_si_{(s+1)}})^2(\sqrt{n})^{-3}\\
&\quad+
3\Delta_{i_{(t-1)}i_t}^2\!+\!
3\sum_{j=1}^{d}\mathbb{E}[(R x_{i_{t-1}}(j)\!-\!R x_{i_t}(j))^2\mathbbm{1}_{\mathcal{A}_t^c}]\\
&\quad+3\alpha_{i_{t-1}}(Q_{\textnormal{M},R})
\\
&
\overset{(a)}{\leq} 6tk^2(\epsilon_{i_1}+\sum_{s=1}^{t-1}\epsilon_{i_si_{(s+1)}})^2(\sqrt{n})^{-3}\\
&\quad+
3\Delta_{i_{(t-1)}i_t}^2+
12\Delta_{i_{(t-1)}i_t}^2(1+3\ln\sqrt{n})(\sqrt{n})^{-3}\\
&\quad+
3\alpha_{i_{t-1}}(Q_{\textnormal{M},R})
\\
&
\leq 6t^2k^2(\epsilon_{i_1}^2+\sum_{s=1}^{t-1}\epsilon_{i_si_{(s+1)}}^2)(\sqrt{n})^{-3}\\
&\quad+
3\Delta_{i_{(t-1)}i_t}^2+
12\Delta_{i_{(t-1)}i_t}^2(1+3\ln\sqrt{n})(\sqrt{n})^{-3}\\
&\quad+
3\alpha_{i_{t-1}}(Q_{\textnormal{M},R}),
\end{aligned}
\end{equation}
where in (a) we use \emph{Lemma} \ref{sG}.

Combining (\ref{eq2}) and (\ref{seq2}), we have
\begin{equation}\label{a1}
\begin{aligned}
&\mathbb{E}[||\hat{x}_{Q_{\textnormal{M},R},i_t}- x_{i_t}||_2^2]\\
&\leq
dt(\epsilon_{i_1}^2+\sum_{s=1}^{t-1}\epsilon_{i_si_{(s+1)}}^2)+
6t^2k^2(\epsilon_{i_1}^2+\sum_{s=1}^{t-1}\epsilon_{i_si_{(s+1)}}^2)(\sqrt{n})^{-3}\\
&\quad\!+\!
3\Delta_{i_{(t-1)}i_t}^2\!+\!
12\Delta_{i_{(t-1)}i_t}^2(1\!+\!3\ln\sqrt{n})(\sqrt{n})^{\!-\!3}\!+\!
3\alpha_{i_{t-1}}(Q_{\textnormal{M},R})
\\
&\overset{(a)}{=}
\frac{24tA_{i_t}\ln\sqrt{n}}{(k-2)^2}+\frac{144t^2k^2A_{i_t}\ln\sqrt{n}}{(k-2)^2n\sqrt{n}}+
3\Delta_{i_{(t-1)}i_t}^2\\
&\quad+
12\Delta_{i_{(t-1)}i_t}^2(1+3\ln\sqrt{n})(\sqrt{n})^{-3}+
3\alpha_{i_{t-1}}(Q_{\textnormal{M},R})
\\
&\overset{(b)}{\leq}
\frac{24tA_{i_t}\ln\sqrt{n}}{(k-2)^2}+\frac{576t^2A_{i_t}}{en}+
3\Delta_{i_{(t-1)}i_t}^2+
\frac{36\Delta_{i_{(t-1)}i_t}^2}{e^{2/3}n}\\
&\quad+
3\alpha_{i_{t-1}}(Q_{\textnormal{M},R})
\\
&\leq\frac{24lA_{i_t}\ln\sqrt{n}}{(k-2)^2}\!+\!c_{t}(n)\frac{A_{i_t}+\Delta_{i_{(t-1)}i_t}^2}{n}+3\alpha_{i_{t-1}}(Q_{\textnormal{M},R}),
\end{aligned}
\end{equation}
where in (a) we use $\epsilon_1=\frac{2\Delta'_1}{k-2}$ and $\epsilon_{(s-1)s}=\frac{2\Delta'_{(s-1)s}}{k-2}$ and in (b) we use $(1+3\ln\sqrt{n})/\sqrt{n}\leq3/e^{2/3}$ and $(\ln\sqrt{n})/\sqrt{n}\leq1/e$ and $A_{i_t}\!=\!\Delta_{i_1}^2\!+\!\sum_{s=1}^{t-1}\Delta_{i_si_{(s+1)}}^2$.

Now, we consider $\beta_{i_t}(Q_{\textnormal{M},R})$. For $\beta_{i_t}(Q_{\textnormal{M},R})$, we have 
\begin{equation}\label{b1}
\begin{aligned}
||\mathbb{E}&[\hat{x}_{Q_{\textnormal{M},R},i_t}]- x_{i_t}||_2^2\\
&=||\mathbb{E}[R\hat{x}_{Q_{\textnormal{M},R},i_t}]-R x_{i_t}||_2^2\\
&\overset{(a)}{\leq}\sum_{j=1}^{d}\mathbb{E}[(R\hat{x}_{Q_{\textnormal{M},R},i_t}(j)-R x_{i_t}(j))\mathbbm{1}_{\bigcup_{s\in[t]}\mathcal{A}_s^c}]^2\\
&\leq\sum_{j=1}^{d}\mathbb{E}[(R\hat{x}_{Q_{\textnormal{M},R},i_t}(j)-R x_{i_t}(j))^2\mathbbm{1}_{\bigcup_{s\in[t]}\mathcal{A}_s^c}]\\
&\leq c_{t}(n)\frac{A_{i_t}+\Delta_{i_{(t-1)}i_t}^2}{n}+3\alpha_{i_{t-1}}(Q_{\textnormal{M},R}),
\end{aligned},
\end{equation}
where  (a) holds by the fact that if $R\in\bigcap_{s\in[t]}\mathcal{A}_s$, then $\hat{x}_{Q_{\textnormal{M},R},i_t}$ is an unbiased estimte of $ x_{i_t}$ by Lemma \ref{cqpe}.

For $t=1$, following a similar method above we have 
\begin{equation}
\begin{aligned}
\mathbb{E}&[||\hat{x}_{Q_{\textnormal{M},R},i_1}- x_{i_1}||_2^2]\\
&\leq
d\epsilon_{i_1}^2+4dk^2\epsilon_{i_1}^2(\sqrt{n})^{-3}+
4(2\Delta_{i_1}^2+d\Delta_{i_1}^{'2})(\sqrt{n})^{-3}
\\
&\leq
\frac{24\Delta_{i_1}^2\ln\sqrt{n}}{(k-2)^2}+(\frac{96}{e}(\frac{k}{k-2})^2+\frac{24}{e^{2/3}})\frac{\Delta_{i_1}^2}{n}
\\
&\leq
\frac{24\Delta_{i_1}^2\ln\sqrt{n}}{(k-2)^2}+154\frac{\Delta_{i_1}^2}{n}
.
\end{aligned}
\end{equation}
Similarly, for $\beta_{i_1}(Q_{\textnormal{M},R})$ we have
\begin{equation*}
||\mathbb{E}[\hat{x}_{Q_{\textnormal{M},R},i_1}]- x_1{i_1}||_2^2\leq 154\frac{\Delta_{i_1}^2}{n}.
\end{equation*}
We complete the proof of the lemma according Lemma \ref{lossyq}.

\subsection{Proof of Corollary \ref{simple}}\label{c1}
We first consider $\alpha_{i_{l}}(Q_{\textnormal{M},R})$ and $\beta_{i_{l}}(Q_{\textnormal{M},R})$. Then, we prove 
\begin{equation*}
\alpha_{i_l}(Q_{\textnormal{M},R})\leq\frac{24D_{i_l}^i\ln\sqrt{n}}{(k-2)^2}+154\frac{D_{i_l}^i}{n},
\end{equation*}
\begin{equation*}
\beta_{i_l}(Q_{\textnormal{M},R})\leq 154\frac{D_{i_l}^i}{n}.
\end{equation*}

We use induction for $l$. If $l=1$, it is obvious.

Now, we suppose
\begin{equation*}
\alpha_{i_{l-1}}(Q_{\textnormal{M},R})\leq\frac{24D_{i_{l-1}}^i\ln\sqrt{n}}{(k-2)^2}+154\frac{D_{i_{l-1}}^2}{n},
\end{equation*}
\begin{equation*}
\beta_{i_{l-1}}(Q_{\textnormal{M},R})\leq 154\frac{D_{i_{l-1}}^i}{n}.
\end{equation*}
By (\ref{a1}) and (\ref{b1}), we have
\begin{equation*}
\begin{aligned}
&\alpha_{i_{l}}(Q_{\textnormal{M},R})\\
&\leq\frac{24lA_{i_l}\ln\sqrt{n}}{(k-2)^2}\!+\!c_{i_l}(n)\frac{A_{i_l}+\Delta_{i_{(l-1)}i_l}^2}{n}+3\alpha_{i_{l-1}}(Q_{\textnormal{M},R})\\
&\leq\frac{24(lA_{i_l}+3D_{i_{l-1}^i})\ln\sqrt{n}}{(k-2)^2}\!+(\!c_{i_l}(n)\frac{A_{i_l}+\Delta_{i_{(l-1)}i_l}^2}{154}+3D_{i_{l-1}^i})\cdot\frac{154}{n}
\end{aligned}
\end{equation*}
and
\begin{equation*}
\begin{aligned}
&\beta_{i_{l}}(Q_{\textnormal{M},R})\\
&\leq\!c_{i_l}(n)\frac{A_{i_l}+\Delta_{i_{(l-1)}i_l}^2}{n}+3\alpha_{i_{l-1}}(Q_{\textnormal{M},R})\\
&\leq\frac{3D^i_{i_{l-1}}24\ln\sqrt{n}}{(k-2)^2}\!+(\!c_{i_l}(n)\frac{A_{i_l}+\Delta_{i_{(l-1)}i_l}^2}{154}+3D_{i_{l-1}^i})\cdot\frac{154}{n}\\
&\leq(\frac{3nD^i_{i_{l-1}}24\ln\sqrt{n}}{154(k-2)^2}\!+\!c_{i_l}(n)\frac{A_{i_l}+\Delta_{i_{(l-1)}i_l}^2}{154}+3D_{i_{l-1}^i})\cdot\frac{154}{n}\\
&\leq(\frac{3nD^i_{i_{l-1}}}{154}\!+\!c_{i_l}(n)\frac{A_{i_l}+\Delta_{i_{(l-1)}i_l}^2}{154}+3D_{i_{l-1}^i})\cdot\frac{154}{n},
\end{aligned}
\end{equation*}
where in the last step, we use the fact $\frac{24\ln\sqrt{n}}{(k-2)^2}\leq 1$ since $\log k=\lceil \log(2+\sqrt{12\ln n})\rceil$.

Since \begin{IEEEeqnarray}{rCl}
	D_{i_l}^i=&&\max\{l(\Delta_{i_1}^2\!+\!\sum_{s=1}^{l-1}\Delta_{i_si_{(s+1)}}^2), \frac{c_{l}(n)}{154}(\Delta_{i_1}^2\\
	\nonumber && +\!\sum_{s=1}^{l-1}\Delta_{i_si_{(s+1)}}^2+\Delta_{i_{(l\!-\!1)}i_l}^2)+\frac{3nD_{i_{l-1}}^i}{154}\}+3D_{i_{l-1}}^i,
\end{IEEEeqnarray} and and Lemma \ref{lossyq}, we complete the proof.
\end{document}